\documentclass[twoside,11pt]{article}

\usepackage{jmlr2e}
\usepackage{amsmath,amssymb}

\usepackage{hyperref}
\usepackage{array}
\usepackage{booktabs}
\usepackage{algorithmic,xcolor}
\usepackage{algorithm}

\def\lln{\ell_n}
\def\llnm{\ell_n^m}

\def\th{\theta}

\def\Ex{\mathbb{E}}

\def\cone{\mathcal{C}}
\def\tth{\tilde{\theta}}
\def\ths{\th^\star}

\def\hth{\hat{\theta}}

\def\rs{_{rs}}
\def\Y{\mathcal{Y}}
\def\grad{\nabla}
\def\hess{\nabla^2}

\def\bbd{\bar{d}}
\def\bd0{\bar{d} _0}
\def\ef{F(\xi,T)}
\def\bef{\bar{F}(\xi,T)}
\def\eftwo{F^2(\xi,T)}

\firstpageno{1}

\usepackage{lastpage}
\jmlrheading{19}{2018}{1-\pageref{LastPage}}{11/16}{11/18}{16-554}{B{\l}a{\.z}ej Miasojedow and Wojciech Rejchel}
\ShortHeadings{Sparse Ising Model via Penalized MC Methods}{Miasojedow and Rejchel}
\begin{document}

\title{Sparse Estimation in Ising Model
via Penalized Monte Carlo Methods }

\author{\name B{\l}a{\.z}ej Miasojedow \email bmia@mimuw.edu.pl \\
       \addr Institute of Applied Mathematics and Mechanics\\
       University of Warsaw\\
       ul. Banacha 2, 02-097, Warszawa, Poland\\
       and\\
       Institute of Mathematics\\
       Polish Academy of Sciences\\
       ul. {\'S}niadeckich 8, 00-656 Warszawa
       \AND
       \name Wojciech Rejchel \email  wrejchel@gmail.com\\
       \addr Faculty of Mathematics and Computer Science \\
       Nicolaus Copernicus University\\
       ul. Chopina 12/18,
87-100 Toru{\'n}, Poland\\
and\\
\addr Institute of Applied Mathematics and Mechanics\\
       University of Warsaw\\
       ul. Banacha 2, 02-097, Warszawa, Poland\\}

       \editor{Sara van de Geer}
\maketitle

\begin{abstract}%
We consider a  model selection problem in high-dimensional binary Markov random fields. 
The usefulness of the Ising model in studying systems of complex interactions has been confirmed in many papers. The main drawback of this model is the intractable norming constant that makes estimation of parameters  very challenging. 
In the paper we propose a Lasso penalized version of the Monte Carlo maximum likelihood method.
We prove that our algorithm, under mild regularity conditions,  recognizes the true dependence structure of the graph with high probability. The efficiency of the proposed method is also investigated via numerical studies.
\end{abstract}


\section{Introduction}
\label{introduction}

A Markov random field is an undirected graph $(V,E),$ where $V=\{1,\ldots, d\}$ is a set of vertices and $E \subset V \times V$ is a set of edges. The structure of this graph describes conditional independence among subsets of a random vector 
$Y=(Y(1),\ldots,Y(d)),$ where a random variable $Y(s)$ is associated with a vertex
$s \in V.$ Finding interactions between random variables is a central element of many branches of science, for example biology, genetics, physics or social network analysis. 
The goal of the current paper is to recognize the structure of a graph on the basis of a sample consisting of $n$ independent graphs. We consider the high-dimensional setting, i.e. the number of vertices $d$ can be comparable or larger than the sample size $n.$ It is motivated by contemporary applications of Markov random fields in the above-mentioned places, for instance gene microarray data.

The Ising model \citep{Ising25} is an important example of a mathematical model that is often used to explain relations between discrete random variables. In the literature one can find many papers that argue for its effectiveness in recognizing the structure of a graph 
\citep{Ravi10, HofTib09,Guoetal10, Xueetal12, Jalali11}. This model also plays  a key role in our paper. On the other hand, the Ising model is an example of an intractable constant model that is  the joint distribution of $Y$ is known only up to a norming constant and this constant cannot be calculated in practice. 

Thus, there are two main difficulties in the considered model. The first one is the high-dimensionality of the problem. The second one is the intractable norming constant. To overcome the first obstacle we apply a well-known Lasso method \citep{Tibsch96}.   
The properties of this method in model selection are deeply studied in many papers that mainly investigate
linear models or generalized linear models \citep{Bickel09, geerbuhl11, HuangGLM12, Geer2008, YeZhang10, zhaoYu06, zhou09}. However, it is not difficult to find papers that describe properties of Lasso estimators in more complex models, for instance Markov random fields \citep{banerjee08, geerbuhl11, Ravi10, HofTib09,Guoetal10, Xueetal12} that are considered in this paper.

There are many approaches trying to overcome the second obstacle that is the intractable norming constant. For instance, in \citet{Ravi10} one proposes to perform $d$ regularized logistic regression problems. This idea is based on the fact that the norming constant reduces, if one considers the conditional distribution instead of the joint distribution in the Ising model. 
This simple fact is at the heart of the pseudolikelihood approach \citep{Besag74} that is replacing the likelihood (that contains the norming constant) by the product of conditionals (that do not contain the norming constant). This idea is widely applied in the literature \citep{HofTib09,Guoetal10, Xueetal12, Jalali11} to study model selection properties of high-dimensional Ising models. However, this approach works well only if the pseudolikelihood is a good approximation of the likelihood. In general, it depends on the true structure of a graph. Namely, if this structure of the graph is sufficiently simple (examples of different structures can be found in section \ref{simulation}), then the product of conditionals should be close to the joint distribution. However, in practice this knowledge is unavailable. 
Another approach is described in \citet{banerjee08}. It adapts the method that estimates the precision matrix in gaussian graphical models to the binary case.
In the current paper we propose the  approach to the norming constant problem that relates to Markov chain Monte Carlo (MCMC) methods. 
Namely, the norming constant is approximated using the importance  
sampling technique. This method is independent of the unknown complexity of the estimated graph. 
It is sufficient that the size of a sample used in importance sampling is sufficiently large to have good approximation of the likelihood. 

The MCMC method is a well-known approach to overcome the problem with the intractable norming constant in classical (low-dimensional) estimation of graphs. 
For instance, its properties  are investigated in influential papers \citet{GeyerThompson92, Geyer94}.
In the high-dimensional Ising model these algorithms were also studied. 
For example, \citet{Honorio2012} and \citet{atchade2017Perturbed}  analyzed stochastic versions
of proximal gradient algorithms. Both papers derive nonasymptotic bounds between the output of the algorithm and the true minimizer of the cost function. 
However, in the current paper we focus on model selection properties of MCMC methods.
We investigate them 
in the high-dimensional
scenario and compare to the existing methods that are mentioned above.
Model selection for undirected graphical models means finding the existing edges in the ``sparse" graph that is a graph having relatively few edges (comparing to the total number of possible edges $\frac{d(d-1)}{2}$ and the sample size $n$).

The paper is organized as follows: in the next section we describe the Ising model and our approach to the problem that relates to minimization of the penalized MCMC approximation of the likelihood. The literature concerning this topic is also discussed. In section~\ref{mainresults} we state main theoretical results. Details of efficient implementation are given 
in section~\ref{implementation}, while the results of numerical studies are presented in section~\ref{numerical}. The conclusions can be found in section \ref{conclusions}. Finally,  the proofs are postponed to appendices ~\ref{auxresults} and \ref{proofsmain}.

\section{Model description and related works}
\label{model_descr}

In this section we introduce the Ising model and the proposed method. It also contains 
a review of the literature relating to this problem.

\subsection{Ising Model and undirected graphs}
\label{Ising_model}

Let $(V,E)$ be an undirected graph that consists of a set of vertices $V$ and a set of edges $E.$
The random vector $Y=(Y(1),Y(2),\ldots,Y(d)),$ that takes values in $\Y,$
is associated with this graph. In the paper we consider a special case of the Ising model that 
$Y(s) \in \{-1,1\}$ and the joint distribution of $Y$ is given by the formula 
\begin{equation}
\label{ising}
p(y| \ths) = \frac{1}{C(\ths)} \exp \left( \sum_{r<s} \ths_{rs} y(r) y(s) \right),
\end{equation}
where the sum in \eqref{ising} is taken over such pairs of indices  $(r,s) \in \{1,\ldots, d\}^2 $ that $r<s.$ The vector $\ths \in \mathbb{R}^{d(d-1)/2}$ is a true parameter and $C(\ths)$ is a norming constant, i.e. 
\[C(\ths)= \sum_{y\in \Y} \exp \left( \sum_{r<s} \ths_{rs} y(r) y(s) \right)\,.\]  
The norming constant is a finite sum but it consists of $2^d$ elements that makes it intractable even for a moderate size of $d$.

For convenience, we denote $J(y)=\left(
y(r)y(s) \right)_{r<s},$ so 
\[p(y|\ths)=  \frac{1}{C(\ths)} \exp \left[(\ths)' J(y) \right]\,.\] 

\begin{remark}
The model \eqref{ising} is a simplified version of the Ising model, for instance we omit an external field in \eqref{ising}. We have decided to restrict to the model 
containing only parameters $\ths _{rs}$, because interactions between random variables is what we focus on  in the current paper. However, our results can be relatively easy extended.
\end{remark}

The Ising model has the following property: vertices $r$ and $s$ are not connected by an edge (i.e. $\ths _{rs} = 0$) means that variables $Y(r)$ and $Y(s)$
are conditionally independent given the other vertices. Therefore, we recognize the structure of the graph (its edges) by estimating the parameter $\ths.$
Assume that $Y_1, \ldots, Y_n$ are independent random vectors from the model \eqref{ising}. Then the negative log-likelihood is 
\begin{equation}
\label{loglike}
\lln(\theta)=  - \frac{1}{n} \sum_{i=1}^{n}   \th ' J(Y_i) + \log {C(\th)}\,.
\end{equation}
The second term in \eqref{loglike} contains the norming constant so we cannot use \eqref{loglike} to estimate $\ths.$ To overcome this problem one usually replaces the negative log-likelihood by its approximation and estimates $\ths$ using the minimizer of this approximation. In the current 
paper the approximation of \eqref{loglike} is based on Monte Carlo (MC) methods. Suppose that  $h(y)$ is an importance sampling distribution and note that
\begin{equation}
\label{constant}
{C(\th)}=\sum_{y \in \Y} \exp \left[\th ' J(y)\right] =\sum_{y \in \Y} \frac{\exp \left[\th ' J(y)\right]}{h(y)}h(y) =\Ex_{Y\sim h}\frac{\exp \left[\th ' J(Y)\right]}{h(Y)}
\end{equation}
for each $\th .$
An MC approximation of the norming constant is
\begin{equation}
\label{approxim}
\frac{1}{m}\sum_{k=1}^{m}\frac{\exp \left[\th ' J(Y^k)\right]}{h(Y^k)} \,,
\end{equation}
where $Y^1,\ldots, Y^m$ is a  sample drawn from $h$ or, which is more realistic and is considered in the current paper,
$Y^1,\ldots, Y^m$ is a Markov chain with $h$ being a density of its stationary distribution. 
 Thus, the MCMC approximation of \eqref{loglike} is 
\begin{equation}
\label{MCapprox}
\llnm (\th) = - \frac{1}{n} \sum_{i=1}^{n}   \th ' J(Y_i) 
+ \log
\left( \frac{1}{m}\sum_{k=1}^{m}\frac{\exp \left[ \th ' J(Y^k)\right]}{h(Y^k)} \right)\,.
\end{equation}
A natural choice of the importance sampling distribution is $h(y) =p(y|\psi)$ for some parameter $\psi.$ It leads to  
\begin{equation} 
\label{McappEx}
\llnm (\th) =  - \frac{1}{n} \sum_{i=1}^{n}   \th ' J(Y_i) 
+ \log
\left( \frac{1}{m}\sum_{k=1}^{m} \exp \left[ (\th- \psi) ' J(Y^k)\right] \right)
+ \log (C(\psi))\,.
\end{equation}
The last term in \eqref{McappEx}, that contains the unknown constant  $C(\psi),$ does not depend on $\th$, so it can be ignored while minimizing \eqref{McappEx}. 

Our goal is selecting the true model (recognizing edges of a graph) in the high-dimensional setting. It means that the number 
of vertices $d$ can be large. In fact, it can be  greater than the sample size, i.e. $d= d_n \gg n.$ To estimate the vector $\ths$ we use penalized empirical risk minimization.
The natural choice of the penalty would be  the $l_0$-penalty but it makes the procedure nonconvex and computationally expensive even for moderate values of $d.$ To avoid such problems we use the Lasso penalty and minimize a function
\begin{equation}
\label{MCLasso}
\llnm(\th) + \lambda_n^m |\th|_1\,,
\end{equation}
where $|\th|_1 = \sum_{r<s} |\th_{rs}|$ and 
$\lambda_n^m >0$ is a smoothing parameter that is a balance between minimizing 
the MCMC approximation and the penalty. 
We denote the minimizer of \eqref{MCLasso} by $\hth _n^m.$ 
Notice that the function \eqref{MCLasso} that we minimize is convex in $\th,$ because the Lasso penalty as well as the MCMC approximation \eqref{McappEx} are convex function in $\th.$ The latter follows from the fact that the Hessian of $\llnm(\th),$ that is given explicitly in \eqref{2deriv}, is  a weighted covariance matrix with positive weights that sum up to one. 
Convexity of the problem is important from the practical and theoretical point of view. First,
every minimum of a convex function is the global minimum, so there are no local minimum problems. Second, convexity is also utilized in the proofs of the results contained in the paper.
In further parts of the paper we study properties of $\hth _n^m$ in model selection.

\subsection{Related works}
\label{related works}

Model selection in the high-dimensional Ising model is a popular topic and many papers investigating this problem using different methods can be found in the 
literature \citep{banerjee08, bresler08, Lee06, Ravi10,HofTib09, Guoetal10,Xueetal12,  Jalali11}. 
The significant part of them uses  the pseudolikelihood approximation  with the Lasso penalty. For instance, \citet{Ravi10} applies 
it by considering $d$ logistic regression problems. They prove that this algorithm is model selection consistent, if some regularity conditions are satisfied. 
These conditions are similar to the ``irrepresentable conditions" \citep{zhaoYu06} that are sufficient to prove an analogous property in the linear model. 
The pseudolikelihood method with the Lasso as a ``joint'' procedure is proposed in \citet{HofTib09}. Moreover, in the same paper one also proposes an ``exact'' algorithm
that minimizes the negative log-likelihood with the Lasso penalty.
However, this procedure also bases on the pseudolikelihood approximation. Model selection consistency of the latter algorithm has not been studied yet. 
The former procedure has this property that is showed in \citet{Guoetal10} provided that  conditions similar to \citet{Ravi10} are satisfied.  
In \citet{Xueetal12} the Lasso penalty is replaced by the SCAD penalty \citep{FanLi2001} and theoretical properties of this algorithm are studied. In \citet{Jalali11} one replaces restrictive irrepresentable conditions by weaker restricted strong convexity and smoothness conditions \citep{negahban12} and proves model selection consistency of an algorithm that joints ideas from \citet{Ravi10} and \citet{zhangt09}. Namely, it performs $d$ separate logistic regression problems with the forward-backward greedy approach.  The algorithm described in \citet{banerjee08} is also based on the likelihood approximation with the Lasso penalty. However, it does not apply the pseudolikelihood method. Using the determinant relaxation \citep{WainJord06} it treats the problem of model selection in discrete Markov random fields analogously to the continuous case.

In the current paper we apply the MCMC method to overcome the intractable norming constant problem. 
Our experimental study (presented in section \ref{numerical}) confirms 
that estimators based on the MCMC approximation usually perform comparably or better in model selection than their competitors from \citet{banerjee08, HofTib09, Ravi10}.
Our theoretical results are similar to  those described in the previous paragraph that is we prove model selection consistency. But, in general, our assumptions are weaker than their analogs from the above-mentioned papers. 
The detailed comparison is given after Corollary~\ref{maincor} in section~\ref{mainresults}.

Moreover, the advantage of our algorithm is that the MCMC approach allows us to approximate the norming constant with an arbitrary precision. The approximation error of other methods is given by the problem/data. It depends on the unknown structure of a graph and a user cannot improve it. In our approach a user can improve approximation by increasing the length of simulation, using the MCMC algorithm tailored to the problem. 
The obvious drawback of our approach is the need of additional simulations to obtain the MCMC sample. It makes our procedure computationally more complex, but at the same time more accurate in selecting the true model.

\subsection{Notation}
\label{notation}

In further parts of the paper we need few notations. Most of them are collected in this subsection.

For simplicity, we write $\hth$ and $\lambda$ instead of $\hth _n^m$ and $\lambda_n^m,$ respectively.  Besides, we denote the number of estimated parameters in the model by $\bar{d} = d(d-1)/2.$ Nonzero coordinates of $\ths$ are collected in the set $T$, and $T^c$ is a completion of $T.$ Besides, 
$\bar{d}_0 = |T|$ denotes the number of elements of the set $T.$

For a vector $a$ we denote its $l_\infty$-norm by $|a|_\infty = \max\limits_k |a_k|$ and $a^{\otimes 2}=aa'.$ The vector $a_T$ is the same as the vector $a$ on $T$ and zero otherwise. The  $l_\infty$-norm of a matrix $\Sigma$ is denoted by $|\Sigma|_\infty= \max\limits_{k,l} |\Sigma_{kl}|.$ 

Let us consider a Markov chain on space $\mathcal{S}$ with transition kernel $P(x,\cdot\,)$ and stationary distribution $\pi$. We define the Hilbert space $L^2(\pi)$ as a space of functions that $\pi(f^2)<\infty$ and the 
inner product is given as $\langle f,g\rangle=\int_\mathcal{S} f(x)g(x)\pi(dx)$. The linear operator $P$ on $L^2(\pi)$ associated with transition kernel $P(x,\cdot\,)$ is defined as follows
\[Pf(x)=\int_\mathcal{S} f(y)P(x,dy)\,.\]
We say that the Markov chain has a spectral gap $1-\kappa$ if and only if \[\kappa=\sup\{|\rho|\colon \rho\in Spec(P)\setminus \{1\}\}\,,\]
where $Spec(\,\cdot\,)$ denotes the spectrum of an operator in $L^2(\pi)$. For reversible chains the spectral gap property is equivalent to geometric ergodicity of the chain, see \citep{kontoyiannis2009geometric,roberts1997geometric}.

In the paper we focus on the Gibbs sampler for the Ising model. However, theoretical results remain true for other MCMC algorithms as long as the spectral gap property is satisfied.
The random scan Gibbs sampler for the Ising model with a joint distribution  $p(y|\psi)$ is defined as follows: given $Y^{k-1}$, first we sample uniformly index $r$ and we draw $Y^{k}(r)$ from the distribution

\begin{equation}\label{eq:gibbs}
\mathbb{P}(Y^k(r)=1)=\frac{\exp\left\{\psi'J(Y^+)  \right\}}{\exp\left\{\psi'J(Y^+)  \right\}+\exp\left\{\psi'J(Y^-)  \right\}}\,, 
\end{equation}
where $Y^+(s)=Y^-(s)=Y^{k-1}(s)$ for $s\neq r$ and $Y^+(r)=1$, $Y^-(r)=-1$. For $s\neq r$ we set $Y^{k}(s)=Y^{k-1}(s)$.

Suppose that $Y^1,\ldots, Y^m$ is a Markov chain on $\Y$ generated by a random scan Gibbs sampler defined as above. By construction the chain is irreducible, aperiodic and 
 $h(y)=p(y|\psi)$ is the density of its stationary distribution. Therefore, the stationary distribution is defined uniquely and the chain is ergodic for any initial measure $\nu$ with the density $q.$ 
 Moreover, there exists a spectral gap $1-\kappa$, because the state space is finite and the chain is reversible. Actually, $\kappa$ is the second greatest absolute value of eigenvalues of the transition matrix.
We will need three quantities related to this Markov chain :
\begin{equation}
\label{betas}
\beta_1 =\sqrt{ \sum_{y \in \Y} \frac{q^2(y)}{h(y)}}, \quad
\beta_2= \frac{1-\kappa}{1+\kappa}, \quad
M = \max_{y \in \Y} \frac{\exp((\ths)'  J(y))}{  h(y) C(\ths)}\,.
\end{equation}
Roughly speaking, these three values can be viewed as: $\beta_1$ -- how close the initial density is to the stationary one, $\beta_2$ -- how fast the chain ``mixes", $M$ -- how close the importance sampling density is to the true density \eqref{ising}.

\section {Main results}
\label{mainresults}

In this section we state key results of the paper. In the first one (Theorem \ref{main}) we show that the estimation error of the minimizer of the MCMC approximation with the Lasso penalty can be controlled. In the second result (Corollary \ref{maincor}) we prove model selection consistency for the thresholded Lasso estimator \citep{zhou09}.

First, we introduce the cone invertibility factor that plays an important role in investigating properties of Lasso estimators. It is defined analogously to \citet{YeZhang10, HuangGLM12, Cox13} 
that concerns linear regression, generalized linear models and the Cox model, respectively. 
It is also closely related to the compatibility condition \citep{Geer2008} or the restricted eigenvalues \citep{Bickel09}.
Thus, for $\xi>1$ and the set $T$  we define a cone as
\[\cone (\xi,T) = \left\{\th: |\th_{T^c}|_1 \leq \xi |\th_{T}|_1\right\}\,. \]
For a nonnegative definite matrix $\Sigma$ the cone invertibility factor is  
\[
F(\xi, T, \Sigma) = \inf _{0 \neq \th \in \cone (\xi,T)} \frac{ \th ' \Sigma \th}{|\th _T|_1 |\th|_\infty}\,.
\]
Cone invertibility factors of Hessians of two functions are crucial in our argumentation.
The first function is the expectation of the negative log-likelihood \eqref{loglike}, i.e.
\begin{equation}
\label{trueloglike}
\Ex \lln (\th) = - \th ' \Ex J(Y) + \log C(\th)
\end{equation}
and the second one is the MCMC approximation \eqref{MCapprox}.
We denote them as
\begin{equation}
\label{F}
\ef= \inf _{0 \neq \th \in \cone (\xi,T)} \frac{ \th ' \hess \log C(\ths) \th}{|\th _T|_1 |\th|_\infty}
\end{equation}
and 
\begin{equation}
\label{Fbar}
\bef= \inf _{0 \neq \th \in \cone (\xi,T)} \frac{ \th ' \hess \llnm(\ths) \th}{|\th _T|_1 |\th|_\infty} \,,
\end{equation}
respectively. Notice that only the values of $\hess \log C(\th)$ and $\hess \llnm(\th)$ at the true parameter $\ths$ are taken into consideration in \eqref{F} and \eqref{Fbar}.

Now we can state main results of the paper.
\begin{theorem}
\label{main}
Let $\varepsilon > 0, \xi >1$ and $\alpha (\xi)=  2+ \frac{e}{\xi-1}\:.$  If
\begin{equation}
\label{ncond}
n \geq  \frac{8 (1+ \xi)^4 \, \alpha^2 (\xi) \, \bd0 ^2 \, \log(2 \bbd /\varepsilon)}{\eftwo }\,,
\end{equation}
\begin{equation}
\label{mcond}
 m \geq \frac{64 (1+\xi)^4 \, \alpha^2 (\xi) \, \bd0^2  \, M^2\,  \log \left[2 \bbd (\bbd+1)\beta_1 / \varepsilon \right] }{\eftwo  \beta_2} \,,
\end{equation}
then with probability at least $1-4 \varepsilon$ we have the inequality
\begin{equation} 
\label{main_in}
\left|\hth - \ths \right| _\infty \leq \frac{2 e \, \xi \,\alpha(\xi) \lambda}{(\xi +1)
[\alpha(\xi)-2] \ef}\,, 
\end{equation}
where
\begin{equation}
\label{lam}
\lambda = \frac{\xi + 1}{\xi -1}  \max\left( 2 \sqrt{\frac{2 \log(2 \bbd/\varepsilon)}{n}}, 8 M \sqrt{\frac{\log\left[(2 \bbd+1)\beta_1 /\varepsilon\right]}{m \beta_2}}
\right)\,.
\end{equation}
\end{theorem}

\begin{corollary}
\label{maincor}
Suppose that conditions \eqref{ncond} and \eqref{mcond} are satisfied.
Let  $\ths _{min} = \min\limits_{(r,s) \in T} |\ths \rs|$ and $R_n^m$ denote the right-hand side of the inequality \eqref{main_in}.
Consider the Lasso estimator with a threshold $\delta >0$ that is the set of nonzero coordinates of the final estimator is defined as $\hat{T} = \{(r,s):
|\hth \rs |> \delta\}.$ 
If $ \ths _{min}/2> \delta \geq  R_n^m,$ then 
\[
P\left( \hat{T} = T \right) \geq 1- 4 \varepsilon\,.
\]
\end{corollary}

The main results of the paper describe properties of estimators that are obtained by minimization of the MCMC approximation \eqref{MCapprox} with the Lasso penalty. Theorem \ref{main} states that the estimation error of the Lasso estimator can be controlled. Roughly speaking, the estimation error is small, if the initial sample size and the MCMC sample size are large enough, the model is sparse and the cone invertibility factor $\ef$ is not too close to zero. The influence of the model parameters $( n,  d, \bd0)$ as well as Monte Carlo parameters $(m,\beta_1, \beta_2, M )$ on the results are explicitly stated. It is worth to emphasize that 
our results work in the high-dimensional scenario, i.e. the number of vertices $d$ can be greater than the sample size $n$ provided that the model is sparse. Indeed, the condition \eqref{ncond} is satisfied even if $\bar{d} \sim O\left(e^{n^{c_1}} \right), \bd0 \sim 
O(n^{c_2})$ and $c_1+2c_2 <1.$ 
The condition \eqref{mcond}, that relates to the MCMC sample size, is also reasonable.
The number $\beta_1$ depends  on the initial and stationary distributions. In general, 
its relation to the number of vertices is exponential. However, in \eqref{mcond} it appears
with the logarithm. Moreover,  $\beta_1$ is also reduced using so called burn-in time,
i.e. the beginning of the Markov chain trajectory is discarded. 
Next, the number $\beta_2$ is related to the spectral gap
of a Markov chain. Under mild conditions the inverse of $\beta_2$ depends polynomially on $d$, and under strong regularity conditions it can be reduced to $ O(d\log d)$ as in \citet{mossel2013exact}. 
Finally, there is also the number $M$ in the condition \eqref{mcond} that relates to the  distance between the stationary distribution $h(\cdot)$ and $p(\cdot|\ths)$. Stating the explicit relation between $M$ and the model seems to be difficult. However,  the algorithm, that we propose to calculate $\hth ,$ is designed in such a way to minimize the impact of $M$ on the results. 
The detailed implementation of the algorithm is given in section \ref{implementation}.

The estimation error of the Lasso estimator in Theorem \ref{main} is measured in $l_\infty$-norm. Similarly to \citet{Cox13}, it can be extended to the general $l_q$-norm, $q \geq 1.$ We omit it, because \eqref{main_in} is sufficient to obtain the second main result of the paper (Corollary \ref{maincor}). It  states that the thresholded Lasso estimator is model selection consistent, if, additionally to \eqref{ncond} and \eqref{mcond}, the nonzero parameters are not too small and the threshold is appropriately chosen. It is a consequence of the fact, which follows from Theorem \ref{main}, that the Lasso seperates significant parameters from irrelevant ones, i.e. for each $(r,s) \in T$ and $(r',s') \notin T$ we have
$
|\hth \rs | >|\hth _{r's'}  |
$ 
with high probability. However, Corollary \ref{maincor} does not give a way of choosing the threshold $\delta$, because both endpoints of the interval $[R_n^m, \ths _{min}/2] $ are unknown. It is not a surprising fact that has been already observed, for instance, in linear models \citep[Theorem 8]{YeZhang10}. In section~\ref{implementation} we propose a method of choosing a threshold that relates to information criteria.
 
We have already mentioned that there are many approaches to the high-dimensional Ising model.
Now we compare conditions that are sufficient to prove model selection consistency in the current paper to those basing on the likelihood approximation. If we simplify regularity conditions in Theorem \ref{main}, Corollary \ref{maincor} and forget about Monte Carlo parameters in \eqref{mcond}, then we have:
\begin{enumerate}
\item[(a)] the cone invertibility factor condition is satisfied, 
\item[(b)] the sample size should be sufficiently large, that is $n>\bd0 ^2 \log d,$
\item[(c)] the nonzero parameters should be sufficiently large, that is $\ths_{min} > \sqrt{\frac{ \log d}{n}}\:.$ 
\end{enumerate}

In \citet[Corollary 1]{Ravi10} one needs stronger irrepresentable condition in (a).  Their analog of (b) is $n\geq v^3 \log d, $ where $v$ is the maximum neighbourhood size. Since $v$ is smaller than $\bd0$, their condition is less restrictive. 
However, in (c) they require the minimum signal strength to be higher than ours, because it has to be larger than  $\sqrt{\frac{v \log d}{n}}$ as distinct from $\sqrt{\frac{ \log d}{n}}$ in our paper.

Assumptions in  \citet[Theorem 2]{Guoetal10} are stronger than ours. Indeed, they need irrepresentable condition in (a), $\bd0$ in the third power in (b) and additional factor $\sqrt{\bd0}$ in (c).

In \citet[Corollary 3.1 (2)]{Xueetal12} model selection consistency of Lasso estimators is also proved with more restrictive conditions than ours. Namely, they are similar to \citet{Ravi10} and 
\citet{Guoetal10} but $\bd0$ is reduced in the condition (c).
Moreover, they also consider the pseudolikelihood approximation with the SCAD penalty 
and shows that the condition (a) seems to be superfluous in this case, see \citet[Corollary 3.1 (1)]{Xueetal12}. However, using the SCAD penalty they minimize a nonconvex function to obtain an estimator, so 
they have to prove that the computed (local) minimizer is the desired theoretic local solution. Their approach can be viewed as a sequence of weighted Lasso problems, so they need auxiliary Lasso procedures to behave well. Therefore, the irrepresentable condition is assumed \citep[Corollary 3.2]{Xueetal12}.  

The conditions sufficient for model selection consistency that are stated in \citet[Theorem 2]{Jalali11} are comparable to ours but also more restrictive. Instead of the condition (a) they consider a similar requirement called the restricted strong convexity condition. It is completed by the  restricted strong smoothness condition. Moreover, in the the lower bound in the condition (c) they need an additional factor $\sqrt{\bd0}$ as well as the upper bound for $\ths _{min}.$

In the proof of Theorem \ref{main} we use methods that are well-known while investigating properties of Lasso estimators as well as some new argumentation. The main novelty (and difficulty) is the use of the Monte Carlo sample that contains dependent vectors. The first part of our argumentation consists of two steps:
\begin{enumerate}
\item[(i) ] the first step can be viewed as ''deterministic''. We apply methods that were developed in \citet{YeZhang10, HuangGLM12, Cox13} and strongly exploit convexity of the considered problem. These auxiliary results are stated in Lemma \ref{basiclem} 
and Lemma \ref{estim} in the appendix 
\ref{auxresults},
\item[(ii)] the second step is ''stochastic''. We state a probabilistic inequality that bounds the 
$l_\infty$-norm of the derivative of the MCMC approximation \eqref{MCapprox} at $\ths,$ that is 
\begin{equation}
\label{deriv}
\grad \llnm(\ths) = - \frac{1}{n} \sum_{i=1}^{n}    J(Y_i) +
\frac{\sum\limits_{k=1}^{m} w_k(\ths) J(Y^k)}
                 {\sum\limits_{k=1}^{m} w_k(\ths) }
\end{equation}
where
\begin{equation}
w_k(\th) =
\frac{\exp \left[ \th ' J(Y^k)\right]}{h(Y^k)}\:, \quad k=1,\ldots, m\,.
 \label{eq: IS_weights}
\end{equation}
Notice that \eqref{deriv} contains independent random variables $Y_1, \ldots, Y_n$ from the initial sample and the Markov chain $Y^1, \ldots, Y^m$ from the MC sample. Therefore, to obtain the exponential inequalities for the $l_\infty$-norm of \eqref{deriv}, which are given in Lemma \ref{derivative} and Corollary \ref{cor_der} in the appendix \ref{auxresults}, we apply the MCMC theory. In particular, we frequently use the following Hoeffding's inequality
for Markov chains \citep[Theorem 1.1]{Miasojedow14}.
\end{enumerate} 

\begin{theorem}
\label{miasojedow}
Let $Y^1, \ldots, Y^m$ be a  reversible Markov chain with a stationary distribution with a density $h$ and a spectral gap $1-\kappa$ . Moreover, let $g: \Y \rightarrow \mathbb{R}$ be a bounded function  and $\mu = \Ex _{Y \sim h} g(Y)$ be a stationary mean value.
Then for every $t>0, m \in \mathbb{N}$ and an initial distribution $q$
\[
P \left( \left| \frac{1}{m} \sum_{k=1}^m g(Y^k) - \mu \right| >t
\right) \leq 2 \beta_1 \exp \left(- \frac{\beta_2 m t^2}{|g|_\infty ^2}
\right)\, ,\]
where $|g|_\infty=\sup\limits_{y \in \Y} |g(y)|$ and $\beta_1, \beta_2$
are defined in \eqref{betas}.
\end{theorem}

The next part of our argumentation relates to the fact that the Hessian of the MCMC approximation at $\ths,$ that is 
\begin{equation}
\label{2deriv}
\hess \llnm(\ths) = 
\frac{\sum\limits_{k=1}^{m} w_k(\ths) J(Y^k)^{\otimes 2}}
                 {\sum\limits_{k=1}^{m} w_k(\ths)} - 
\left[\frac{\sum\limits_{k=1}^{m} w_k(\ths) J(Y^k)}
                 {\sum\limits_{k=1}^{m} w_k(\ths) } \right]^{\otimes 2}\, ,
\end{equation} 
is random variable. Similar problems were considered in several papers investigating properties of Lasso estimators in the high-dimensional Ising model \citep{Ravi10, Guoetal10, Xueetal12} or the Cox model \citep{Cox13}. We overcome this difficulty by bounding from below the cone invertibility factor $\bef$ by nonrandom $\ef .$ Therefore, we need to prove that 
Hessians of  
\eqref{trueloglike} and \eqref{MCapprox} are close. It is obtained again using the MCMC theory in Lemma \ref{matrix_diff} and Corollary \ref{Fdiff} in the appendix 
\ref{auxresults}.
 
Finally, the proofs of Theorem \ref{main} and Corollary \ref{maincor} are stated in the appendix \ref{proofsmain}.

\section{Details of implementation}\label{implementation}
In this section we describe in details practical implementation of the algorithm analyzed in the previous section.

The solution of the problem \eqref{MCLasso} depends on the choice of $\lambda$ in the penalty term and the parameter $\psi$ in the instrumental distribution. 
Finding ``optimal'' $\lambda$ and $\psi$ is difficult in practice. 
To overcome this problem  we compute a sequence of minimizers  $(\hat\theta_i)_i$ such that 
$\hth _i$ corresponds to $\lambda = \lambda_i$ and the sequence $(\lambda_i)_i$ is decreasing. In the second step we use the Bayesian Information Criterion (BIC) to choose 
the final penalty $\lambda$. More precisely, we start with the greatest value $\lambda_0$ for which the entire vector $\hat\theta$ is zero. 
For each value of $\lambda_i,$ $i \geq 1,$ we set $\psi=\hat\theta_{i-1}$ 
and use the MCMC approximation \eqref{MCLasso} with $Y^1,\dots,Y^m$ given by a Gibbs sampler with the stationary distribution $p(\cdot | \hth _{i-1})$.
This scheme exploits warm starts and leads to  more stable algorithm. 
Next, the estimator $\hth$ is chosen using  BIC that is a popular method of choosing $\lambda$ in the literature, for instance in \citet{Xueetal12}. 
Notice that the function $\ell^m_n(\th)$ is  convex, so we can use proximal gradient algorithms to compute $\hat\theta_i$ as a solution of
\eqref{MCLasso} for a given $\lambda_i$ and $\psi$. Precisely, we use the FISTA algorithm with backtracking from \citet{FISTA}. The whole procedure is summarized in Algorithm~1.

\begin{algorithm}[H]\caption{MCMC Lasso for Ising model}
\begin{algorithmic}
 \STATE Let $\lambda_0>\lambda_1>\cdots>\lambda_{100}$ and $\psi=0$.
 \FOR{$i=1$ to $100$ }
 \STATE Simulate $Y^1,\dots,Y^m$ using a Gibbs sampler with the stationary distribution $p(y|\psi)$.
 \STATE Run the FISTA algorithm to compute $\hat\theta_i$ as 
 \[\arg \min\limits_{\th} \,\{ \ell_n^m(\theta)+\lambda_i|\theta|_1\}\,.\]
 \STATE Set $\psi=\hat\theta_i.$
  \ENDFOR
 \end{algorithmic}
 Next, set $\hth = \hth _{i^*}$, where
\[
 i^*=\arg \min_{1 \leq i \leq 100} \left \{n\llnm (\hat{\th} _i)+\log(n)\Vert \hat{ \theta_i} \Vert_0\right\}
\]
and $\Vert\theta\Vert_0$ denotes the  number of non-zero elements of  $\theta$.
\end{algorithm}

In Algorithm~1 we use $100$ values of $\lambda$ uniformly spaced on the log scale, starting from the largest $\lambda$, which corresponds to the empty model.
We use $m=10^3 d$ iteration of the Gibbs sampler. To compute $\llnm (\hat{\th} _i)$ for $i=1,\ldots,100$ in BIC we generate one more sample of the size $m=10^4 d$ using the Gibbs sampler with the stationary distribution  $p(\cdot|\hat \theta_{50}).$

The important property of our implementation is that the chosen $\psi = \hth _{i-1}$ is usually close to $\hth _i,$ because differences between consecutive $\lambda_i$'s are small. 
In our studies the final estimator $\hth$ is the element of the sequence $(\hth _i)_i$ that 
recognizes the true model in the best way, i.e. it minimizes the MCMC approximation and is sparse simultaneously. One believes that the final estimator $\hth=\hth _i$ is close to $\ths,$ therefore the chosen $\psi= \hth _{i-1}$ should  be also 
similar to $\ths$ that makes $M$ in \eqref{betas} close to one. 
Finally, notice that conditionally on the previous step our algorithm fits to the framework described in subsection \ref{Ising_model}.

Note that in the first iteration in Algorithm~1 we use an uniform distribution on $\{-1,+1\}^d$ as an instrumental distribution and we can use i.i.d sample $Y^1,\dots,Y^m$. Therefore, for 
$\lambda_1$
we get $\beta_1=\beta_2=1$. When we compute estimators for $\lambda_i$ with $i\geq 2$  we use the last sample from the previous step as an initial point, so since the Markov chain generated by the Gibbs sampler is ergodic its initial distribution should be close to 
$p(y|\hat \theta_{i-2})$. Therefore, even without the burn-in time  $\beta_1$ should be approximately equal to the $L^2$-distance between $p(y|\hat \theta_{i-1})$ and $p(y|\hat \theta_{i-2})$, which is small because differences between
consecutive $\lambda_i$'s are small. So, after discarding initial iterations  $\beta_1$ is further reduced. Our choice of stationary distributions also leads to relatively small variances of importance sampling weights given in  \eqref{eq: IS_weights}.

The bounds on $\beta_2$ are challenging problem itself and the sharp bounds are available only in very specific cases. Due to that, we do not have the explicit control on $\beta_2$. However, in our procedure
the stationary distributions are typically given by sparse Ising models and by \citet{mossel2013exact} the spectral gap depends mostly on the number of existing edges in the graph. Therefore $\beta_2$ should not vanish to rapidly with dimensionality of the problem.

Finally, the thresholded estimator is obtained using the Generalized Information Criterion (GIC).
For a prespecified set of thresholds $\Delta$ we calculate  
\[
\delta^* =\arg \min_{\delta \in \Delta} \left \{n\llnm (\hat{\th}^\delta)+\log(\bar d)\Vert \hat{\th} ^\delta \Vert_0\right\}\;,
\]
where $\hat{\th}^\delta$ is the Lasso estimator $\hth$ after thresholding with the level $\delta$. To compute $\llnm (\hat{\th} ^\delta)$ for $\delta \in \Delta$ in GIC we generate the last sample of the size $m=10^4 d$ using the Gibbs sampler with the stationary distribution  $p(\cdot|\hat \theta).$ To find the optimal threshold we apply GIC that uses larger penalty than BIC. Choosing the threshold in this way should be better in model selection and is recommended, for instance, in \citet{pokmiel:15}.

In the rest of this section we discuss computational complexity of our method.
The computational cost of a single step of the Gibbs sampler is dominated by computing probability \eqref{eq:gibbs}, which is of the order $O(r)$, where $r$ is the maximal degree of vertices in a graph related to the stationary distribution. 
In the paper we focus on estimation of sparse graphs, so the proposed $\lambda _i$'s have to be sufficiently large to make $\hth _i$'s sparse. Therefore, the degree $r$ is rather small and the computational cost of generating
$Y^1,\dots,Y^m$ is of order $O(m)$. Next, we need to compute $\ell^m_n(\th)$ and its gradient. For an arbitrary Markov chain the cost of these computation is of the order $O(d^2m)$. But when we use single site updates as in the Gibbs sampler
we can reduce it to $O(dm)$ by remembering which coordinate of $Y^k$ where updated. Indeed, if we know that only the $r$ coordinates  are updated in the step $Y^k\to Y^{k+1},$ then
\begin{multline*}
\th ' J(Y^{k+1})=\th' J(Y^k)+ \sum_{s\colon s<r} \left(\th_{sr}[Y^{k+1}(s)Y^{k+1}(r)- Y^{k}(s)Y^{k}(r)]\right)
\\+\sum_{s\colon s>r} \left(\th_{rs}[Y^{k+1}(s)Y^{k+1}(r)- Y^{k}(s)Y^{k}(r)]\right)\,.
\end{multline*}

Finally, it is well-known that FISTA \citep{FISTA} achieve accuracy $\epsilon$ in $O({\epsilon}^{\frac{1}{2}})$ steps. So, the total cost of computing the solution for single $\lambda_i$ with precision $\epsilon$
is of order $O({\epsilon}^{\frac{1}{2}}md)$. The further reduction of the cost can be obtained  using sparsity of $\hat\theta_{i-1}$ in computing $\ell^m_n(\th)$ and its gradient, and  introducing active variables inside the FISTA algorithm.


\section{Numerical experiments}\label{numerical}
In this section we present efficiency of the proposed method via numerical studies. First we compare our method to three algorithms, which we have mentioned previously,
using simulated data sets. Next we apply our method to the real data example.

\subsection{Simulated data}\label{simulation}

\begin{figure}[htb]
 \centering
 \includegraphics[width=\textwidth]{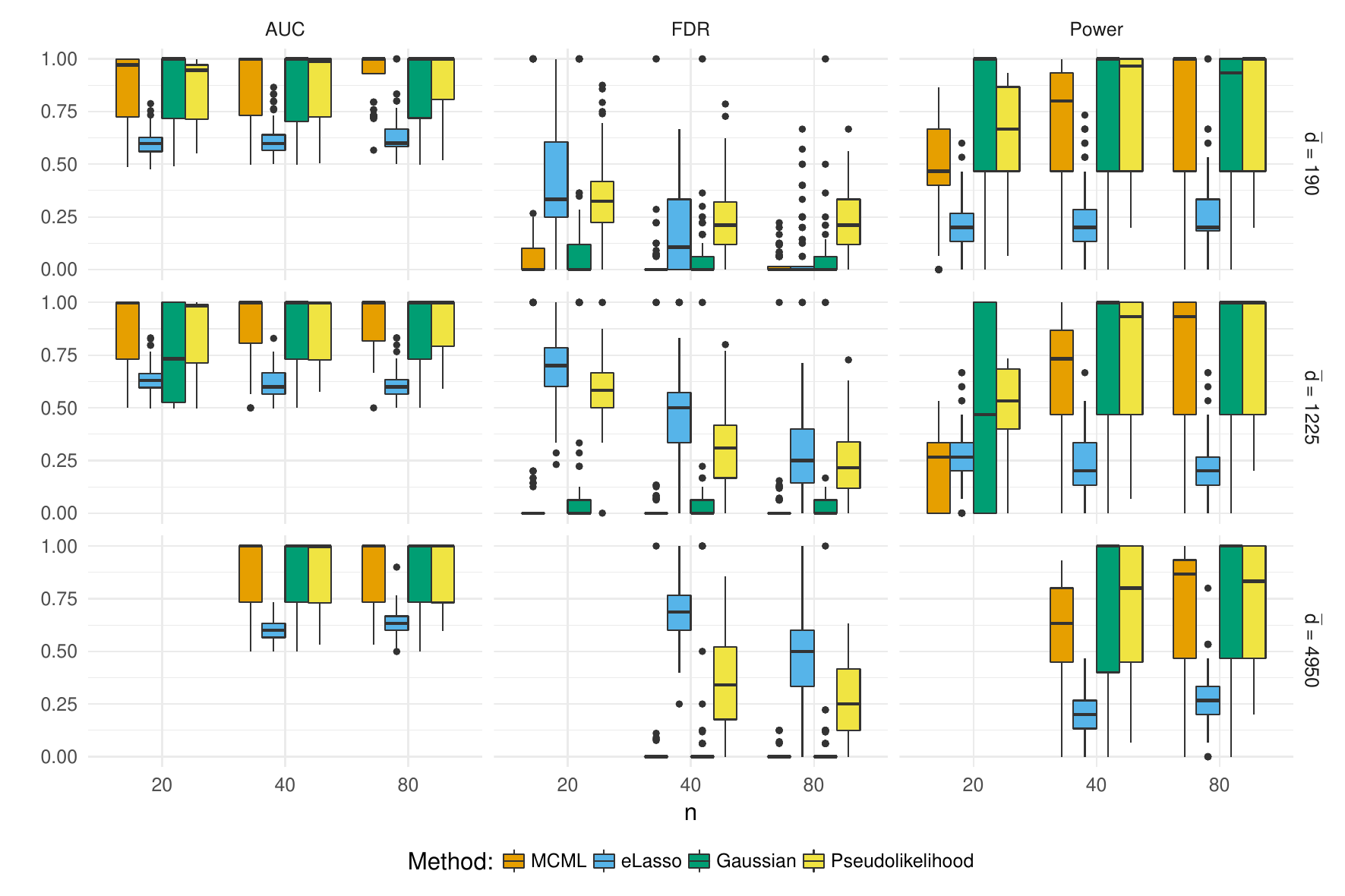}
 \caption{The results for $M1$ model}
 \label{fig:M1}
\end{figure}

To illustrate the performance of the proposed method we simulate  data sets in two scenarios:
\begin{itemize}
 \item[M1]  The first $6$ vertices are correlated, while the remaining vertices are independent: $\ths _{rs}=\pm 2$ for $r<s$ and $s=2,3,4,5,6$, other $\ths _{rs}=0$.
 Thus, the model dimension in this problem is $15$. The signs are chosen randomly.
 \item[M2]  The first $20$ vertices have the ``chain structure'', and the rest are independent: $\ths _{r-1,r}=\pm 1$ for $r\leq 20$. Again the signs are chosen randomly 
 and the model dimension is $19.$
 \end{itemize}
The model $M1$ corresponds to a dense structure on small subset of vertices. The model $M2$ is a simple structure which involves relatively large subset of vertices.


\begin{figure}[htb]
 \centering
 \includegraphics[width=\textwidth]{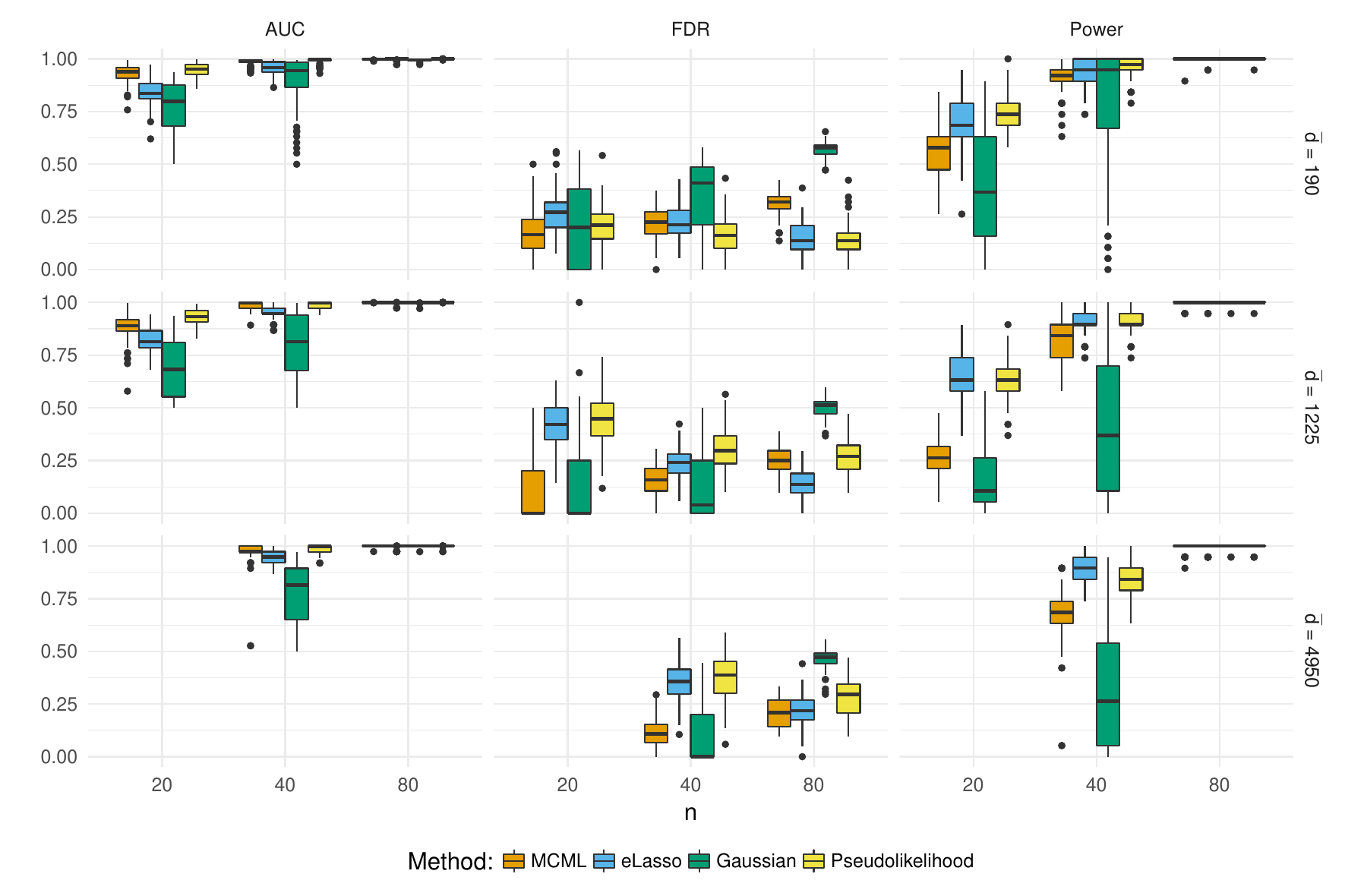}
 \caption{The results for $M2$ model}
 \label{fig:M2}
\end{figure}

We consider the following cases: $d=20,50,100$. So, the considered  number of possible edges (parameters of the model) is $\bar d = 190, 1225, 4950$, respectively. 
For $\bar d=190,1225$ we use  $n=20,40,80$ and for $\bar d=4950$ we use $n=40,80$.

For each configuration of the model, the number of vertices $d$ and the number of observation $n$ we sample $100$ replications of data sets. 
For sampling a data set we use a final configuration of independent Gibbs samplers of the length $10^6$.

In simulation study we compare our methods to the following methods: the pseudolikelihood approach,  \textsc{eLasso} proposed by \citet{Borkulo2014} and the Gaussian approximation
from \citet{banerjee08}. The \textsc{eLasso} method is based on separate logistic models and it is similar to the method proposed by \citet{Ravi10}. For all methods except
\textsc{eLasso} we use two stage procedures, which are analogous to Algorithm~1. Namely, we choose $\lambda$ by BIC  and in the second step we choose the threshold $\delta$ by GIC. 
 For the Gaussian approximation we use the same approximate likelihood as in \citet{Viallon2013}. For \textsc{eLasso} for every node we use BIC as in \citet{Borkulo2014}.

In the comparison we use the following measures of the accuracy. First, to observe the ability of methods to separate true and false edges we compute AUC, where the ROC curve is computed as the threshold $\delta$ varies.
The estimates are also compared using the false discovery rate (FDR) and the ability of recognizing true edges (denoted by
``Power''). The results are summarized in Figure~\ref{fig:M1} and Figure~\ref{fig:M2} for models $M1$ and $M2,$  respectively.

In the first model we can observe that our algorithm and the Gaussian approximation work very well and comparably. The latter has slightly larger power, but at the price of slightly larger FDR. The dominance of these two methods over the pseudolikehood estimator and \textsc{eLasso} is evident. The pseudolikelihood method finds well true edges, but is not able to discard false edges. \textsc{eLasso} works very poorly in this model.  

The second model has a simple structure, so both methods based on the pseudolikelihood approach work much better than in $M1.$ The accuracy of the Gaussian approximation is weak in this model. It has substantial problem with finding true edges, especially when $d$ is large and $n$ small. For $n=20,40$ our estimator has relatively small FDR and large Power. In this model we can observe that FDR increases as $n$ increases for the Gaussian approximation  and for $n=80$ it reaches about 0.5. The MCML approximation has also ``increasing FDR'', but this behaviour is less conspicuous.
In fact, for $n=80$ FDR of our algorithm is still comparable to those of \textsc{eLasso} and Pseudolikelihood.

Summarizing, it is difficult to indicate the winner algorithm. We can observe that the quality of our algorithm  in selecting  the true model is satisfactory. Moreover, only this procedure works on a good level in both models and avoids making noticeable  mistakes.  The Gaussian approximation works well in  $M1$, but seems to be the weakest in  $M2$. The \textsc{eLasso}  completely fails in  $M1,$ but its quality in $M2$    is high, especially for  large sample sizes. The pseudolikelihood approach in all examples well separates true and false edges, has good power but in comparison with other methods 
its FDR is too high, so models that it chooses contains many irrelevant edges.

 Clearly, by construction the computational cost of our method is larger than
its three competitors. However, the whole time that is needed to compute our estimator is reasonable. 
For instance, for a single data set and $n=40$ computing the estimator on 3.4 GHZ CPU takes about 15 seconds for $\bar{d} = 190,$ 60 seconds for $\bar{d} = 1225$ 
and 5 minutes for $\bar{d} = 4950.$ Since our algorithm uses only sufficient statistics the dependence on $n$ of its computational cost is negligible. Thus, the computational times for $n=80$ are almost the same.

\subsection{CAL500 dataset}\label{CAL500}

\begin{figure}[htb]
 \centering
 \includegraphics[width=1.2\textwidth]{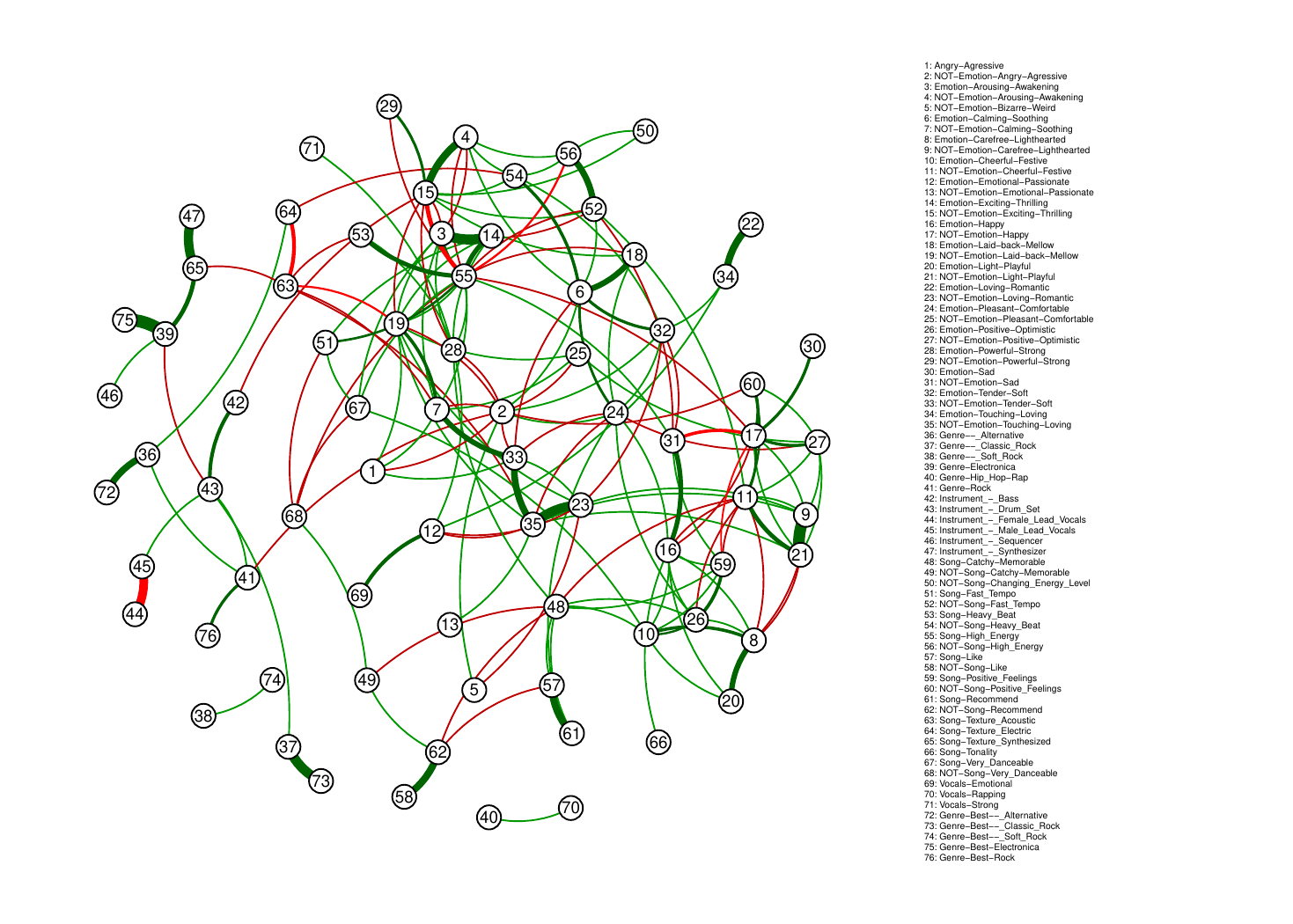}
 \caption{The obtained graph for CAL500 dataset. The width of the edges corresponds to the magnitude of $|\hat{\theta}^{\delta^*}_{rs}|$. The edges with higher absolute value are wider.
 The color denotes the sign of $\hat{\theta}^{\delta^*}_{rs}$: green -- positive, red -- negative.
 To improve clarity we do not show  edges corresponding to $|\hat{\theta}^{\delta^*}_{rs}|<0.001$.}
 \label{Fig:cal500}
\end{figure}
We apply our method to ``CAL500'' dataset \citep{Turnbull2008}. 
Working with a real data set we consider the Ising model \eqref{ising} with a linear term (an external field). This modification is motivated by the fact that in practice marginal probabilities of $Y(s)$ being $+1$  or $-1$ are unknown and should be estimated. The adaptation of our algorithm to this case is straightforward.

For model selection in the Ising model there are no natural measures of the quality of estimates. One would try to compare the prediction ability of obtained estimators, but prediction for the Ising model is challenging itself and results will be biased  by the method used to approximate predicted states. Moreover, all considered methods optimize different loss functions, so these loss functions also cannot be used to the honest comparison of the methods. Due to that, we decide to show only the results of our method for the real data example.

 The considered data set consists of $174$ binary features and $68$ numeric features for $502$ songs. We skipped 
the numeric features and apply our method to find the dependence structure between labels. These 
labels concerning genre, mood or instrument are annotated to songs.
We run our algorithm analogously to the case of simulated data and as the result
we obtain a sparse graph with $181$ edges, see Figure~\ref{Fig:cal500}. We observe that founded edges are rather intuitive. For instance, among the most positively correlated labels we have labels denoted by 3 and 14 (,,Emotion-Arousing-Awakening'' and ,,Emotion-Exciting-Thrilling''),
57 and 61 (,,Song-Like'' and ,, Song-Recommend'') or 22 and 34 (,,Emotion-Loving-Romantic'' and ,,Emotion-Touching-Loving''). On the other hand, the most negatively correlated labels are: 44 and 45 (,,Instrument-Female Lead Vocals'' and ,,Instrument - Male Lead Vocals'') or 63 and 64 
(,,Song-Texture Acoustic'' and ,,Song-Texture Electric'').

\section{Conclusions}\label{conclusions}

In the  paper we consider a problem of structure learning for binary Markov random fields. 
We base estimation of model parameters on the Lasso penalized Monte Carlo approximation
of the likelihood. In the theoretical part of the paper  we show that the proposed procedure reveals the true dependence structure with high probability.
The regularity conditions that we need are not restrictive and are weaker than assumptions used in the other approaches based on the likelihood approximation. 
Moreover, the theoretical results are completed by numerical experiments.
They confirm that the MCMC approximation is able to find the true model in a satisfactory way 
and its quality is comparable or higher than competing algorithms. 

The results of the current paper can be easily extended to other discrete Markov random 
fields. There are also some non-trivial issues that are not discussed in the paper, for instance investigating the model \eqref{ising} with predictors (covariates). The evaluation of the prediction error of the procedure is also a difficult task. Clearly, these  problems need detailed studies.

\acks{We would like to thank the associate editor and two reviewers for their comments that have improved the paper. B{\l}a{\.z}ej Miasojedow and Wojciech Rejchel are supported by Polish National Science Center grants no. 2015/17/D/ST1/01198 and no. 2014/12/S/ST1/00344, respectively. }


\appendix
\section{Auxiliary results}
\label{auxresults}

In this section we formulate lemmas that are needed to prove main results of the paper. The roles that they play are described in detail at the end of section \ref{mainresults}.
The first lemma is borrowed from \citet[Lemma 3.1]{Cox13}.

\begin{lemma}
\label{basiclem}
Let $\tth = \hth - \ths$, $z^* = |\grad \llnm(\ths)|_\infty$ and 
\[D(\hth, \th)= (\hth-\th)' \left[\grad \llnm(\hth) - \grad\llnm (\th) \right]\,.\] Then 
\begin{equation}
\label{basic}
(\lambda - z^*) |\tth_{T^c}|_1 \leq D(\hth, \ths) + (\lambda - z^*) |\tth_{T^c}|_1 
\leq (\lambda + z^*) |\tth_{T}|_1\, .
\end{equation}
Besides, for arbitrary $\xi >1$ on the event 
\begin{equation}
\label{omega1}
\Omega_1=\left\{ |\grad \llnm (\ths)|_\infty \leq \frac{\xi -1}{\xi +1} \lambda \right\}
\end{equation} 
 the random vector $\tth$ belongs to the cone $\cone (\xi, T).$ 
\end{lemma}
\begin{proof}
The proof is the same as the proof of \citet[Lemma 3.1]{Cox13}. It is quoted here to make the paper complete.\\
Convexity of the MCMC approximation $\llnm(\th)$ easily implies the first inequality in \eqref{basic}. The same property combined with convexity of the Lasso penalty gives us
that zero has to belong to the subgradient of \eqref{MCLasso} at the minimizer $\hth,$ i.e.
\begin{equation}
\label{kkt}
\left\{
\begin{array}{cc}
\grad \rs \llnm(\hth) = -\lambda {\rm sign} (\hth _{rs}), & {\rm if} \; \hth \rs \neq 0\\
|\grad \rs \llnm(\hth)| \leq \lambda, & {\rm if} \; \hth \rs = 0\,,
\end{array}
\right.
\end{equation}
where we use $\grad  \llnm(\th)= \left(\grad \rs \llnm(\hth) \right)_{r<s}$ and ${\rm sign}(t)=1$ for $t>0,$ ${\rm sign}(t)=-1$ for $t<0,$ ${\rm sign}(t)=0$ for $t=0.$ 
Using \eqref{kkt} and properties of the $l_1$-norm we obtain that
\begin{align*}
D(\hth, \ths)  &= \sum_{(r,s) \in T} \tth \rs \grad \rs \llnm(\ths+\tth) +
\sum_{(r,s) \in T^c} \hth \rs \grad \rs \llnm(\ths+\tth) - \tth ' \grad \llnm (\ths) \\
&\leq \lambda \sum_{(r,s) \in T} |\tth \rs|  - \lambda \sum_{(r,s) \in T^c} |\hth \rs|
+|\tth|_1 z^*\\
&\leq \lambda |\tth _T|_1 - \lambda |\tth _{T^c}|_1 + z^* |\tth _T|_1 +z^* |\tth _{T^c}|_1 \\
&= (\lambda +z^*) |\tth _T|_1 + (z^*- \lambda) |\tth _{T^c}|_1\, .
\end{align*}
Thus, the second inequality in \eqref{basic} is also established. To prove the last claim of the lemma notice that on the event $\Omega _1$ we obtain
from \eqref{basic} 
\[
|\tth _{T^c}|_1 \leq \frac{\lambda+z^*}{\lambda-z^*} |\tth _T|_1 \leq \xi |\tth _T|_1\, .
\] 
\end{proof}

The second lemma is an adaptation of \citet[Theorem 3.1]{Cox13} to our problem.

\begin{lemma}
\label{estim}
Let $\xi >1 .$  Moreover,
let us denote $\tau = \frac{(\xi +1) \bd0 \lambda }{ \bef }$ and an event
\begin{equation}
\label{omega2}
\Omega_2=\left\{\tau < e^{-1} \right\}\,.
\end{equation}
Then $\Omega_1 \cap \Omega_2 \subset A,$ where
\begin{equation}
\label{estim1}
A= \left\{|\hth - \ths| _\infty \leq \frac{2 \xi e^\eta \lambda}{(\xi+1) \bef } \right\}\,, 
\end{equation}
where $\eta < 1 $ is the smaller solution of the equation $\eta e ^{- \eta} = \tau.$
\end{lemma}

\begin{proof} Suppose we are on the event $\Omega_1 \cap \Omega_2.$ Denote again $\tth = \hth - \ths$ and notice that $\th = \frac{\tth}{|\tth|_1} \in \cone (\xi, T)$ by Lemma \ref{basiclem}.
Consider the function
\[
g(t) = \th ' \grad \llnm(\ths + t \th) -  \th ' \grad \llnm(\ths)
\]
for each $t \geq 0.$ This function is nondecreasing, because  $\llnm(\cdot) $ is convex.
Thus, we obtain 
$
g(t) \leq g(|\tth|_1)
$
for every $t \in (0,|\tth|_1) $.
On the event $\Omega_1$ and from Lemma \ref{basiclem} we have that 
\begin{equation}
\label{estim2}
\th ' \left[ \grad \llnm(\ths + t \th) -  \grad \llnm(\ths)
\right] + \frac{2 \lambda}{\xi +1} |\th _{T^c}|_1 \leq
\frac{2 \lambda \xi}{\xi +1} |\th _{T}|_1 \,.
\end{equation}
In further argumentation we consider all nonnegative $t$ satisfying \eqref{estim2} that is an interval $[0,\tilde{t}]$ for some $\tilde{t} >0$.
Proceeding similarly to the proof of \citet[Lemma 3.2]{Cox13} we obtain 
\begin{equation}
\label{estim3}
t \th '  \left[ \grad \llnm(\ths + t \th) -  \grad \llnm(\ths)
\right] \geq t^2 \exp (- \gamma_{t \th}) \th ' \hess \llnm(\ths) \th\,,
\end{equation}
where $\gamma_{t \th }= t \max_{k,l} |\th ' J (Y^k) -\th ' J (Y^l) | \leq 2t,$ because $J(Y^k)= \left( Y^k (r) Y^k (s)\right) \rs$  and $|\th |_1 =1$.
Therefore, the right-hand side in \eqref{estim3} can be lower bounded by 
\begin{equation}
\label{estim4}
t^2 \exp (-2t) \th ' \hess \llnm(\ths) \th\,.
\end{equation}
Using the definition of $\bef ,$ the fact that $\th \in \cone (\xi,T), $ the bound \eqref{estim4}  and \eqref{estim2} we obtain
\begin{align*}
t \exp(-2t) \frac{\bef  |\th _T|_1^2}{\bd0} &\leq t \exp(-2t) \th ' \hess
\llnm(\ths) \th \\
&\leq \th '  \left[ \grad \llnm(\ths + t \th) -  \grad \llnm(\ths)
\right] \\
 &\leq \frac{2 \lambda \xi}{\xi +1} |\th _{T}|_1 -  \frac{2 \lambda}{\xi +1} |\th _{T^c}|_1\\
&\leq \lambda (\xi +1) |\th _T|_1^2 /2\,.
\end{align*}
So, every $t$ satisfying \eqref{estim2} has to fulfill the  inequality 
$
2t \exp(-2t) \leq \tau
$.
In particular, $2\tilde{t} \exp(-2\tilde{t}) \leq \tau.$ We are on $\Omega_2,$ so it implies that $2 \tilde{t} \leq \eta,$ where $\eta $ is the smaller solution of the equation $\eta \exp(-\eta) =\tau.$
We know also that $|\tth|_1 \leq \tilde{t}$, so 
\begin{align*}
|\tth|_1 \exp(-\eta) &\leq \tilde{t} \exp(-2\tilde{t}) \leq
 \frac{\tilde{t} \exp(-2\tilde{t}) \th ' \hess \llnm(\ths) \th}{ \bef |\th _T|_1 |\th|_\infty} \\
 &\leq\frac{\th '  \left[ \grad \llnm(\ths + \tilde{t} \th) -  \grad \llnm(\ths)
\right]}{ \bef |\th _T|_1 |\th|_\infty} \\
&\leq   \frac{2 \lambda \xi}{(\xi +1) \bef |\th|_\infty}\, ,
\end{align*}
where we have used bounds \eqref{estim4} and \eqref{estim2}. Using the equality 
$|\th|_\infty = \frac{|\tth|_\infty}{|\tth|_1},$ we finish the proof.
\end{proof}

\begin{lemma}
\label{derivative}
For every natural $n,m$ and positive $t$
\begin{equation*}
 P \left( | \grad \llnm(\ths)|_\infty \leq t
\right) \geq 1 - 2 \bbd \exp \left( -nt^2/8 \right)
- \beta_1 \exp \left(  -  \frac{m \beta_2  }{4 M^2}                     \right)  -2 \bbd \beta_1 \exp \left(  -  \frac{ t^2 m \beta_2  }{64 M^2} \right)\,.
\end{equation*}
\end{lemma}
\begin{proof}
We can rewrite $\grad \llnm(\ths)$ as
\begin{equation}
\label{der1}
\grad \llnm(\ths) = - \left[\frac{1}{n} \sum_{i=1}^{n}    J(Y_i) - \frac{\grad C(\ths)}{C(\ths)} \right] +
\frac{\frac{1}{m} \sum\limits_{k=1}^{m} w_k(\ths) \left[ J(Y^k) - \frac{\grad C(\ths)}{C(\ths)} \right]}
                 {\frac{1}{m} \sum\limits_{k=1}^{m} w_k(\ths) }
\end{equation}
Notice that the first therm in \eqref{der1} depends only on the initial sample $Y_1, \ldots, Y_n$ and is an average of i.i.d random variables. The second term depends only on the MCMC sample $Y^1, \ldots, Y^m.$ We start the analysis with the former one.
Using Hoeffding's inequality we obtain for each natural $n$, positive $t$ and a pair of indices $r<s$
$$P \left(\left| \frac{1}{n} \sum_{i=1}^{n}    J \rs (Y_i) - \frac{\grad \rs C(\ths)}{C(\ths)} \right| > t/2
\right) \leq 2 \exp \left( -nt^2/8 \right).
 $$
Therefore, by the union bound we have   
\begin{equation}
\label{deriv_form}
P \left(\left| \frac{1}{n} \sum_{i=1}^{n}    J  (Y_i) - \frac{\grad  C(\ths)}{C(\ths)} \right|_\infty > t/2
\right) \leq 2 \bbd \exp \left( -nt^2/8 \right)\,.
\end{equation}
Next, we investigate the second expression in \eqref{der1}. Its denominator is an average that depends on the Markov chain. To handle it we can use Theorem \ref{miasojedow} in section \ref{mainresults}. 
Notice that $\frac{\exp[(\ths)'J(y)]}{h(y)} \leq M C(\ths)$ for every $y$ and 
$\Ex_{Y \sim h} \frac{\exp\left[(\ths)'J(Y)\right]}{h(Y)} = C(\ths).$ Therefore, for every $n,m$ 
\begin{equation}
\label{der2}
P\left( \frac{1}{m} \sum\limits_{k=1}^{m} w_k(\ths) \geq C(\ths)/2
\right)  \geq 1- \beta_1 \exp \left(  -  \frac{m \beta_2  }{4 M^2}                     \right)\,.
\end{equation}
Finally, we bound the $l_\infty$-norm of the numerator of the second term in \eqref{der1}. We fix a pair of indices $r<s.$ It is not difficult to calculate that 
\[
\Ex _{Y \sim h} \left[\frac{\exp \left[(\ths)'J(Y)\right]}{h(Y)} \left( J \rs (Y) - \frac{\grad \rs C(\ths)}{C(\ths)} \right) \right] =0
\]
and for every $y$
\[
\frac{\exp((\ths)'J(y)}{h(y)} \left| J \rs (y) - \frac{\grad \rs C(\ths)}{C(\ths)}
\right| \leq 2 M C(\ths)\, .
\]
From Theorem \ref{miasojedow} we obtain for every positive $t$
\[
\left|\frac{1}{m} \sum\limits_{k=1}^{m} w_k(\ths) \left[ J \rs (Y^k) - \frac{\grad \rs C(\ths)}{C(\ths)} \right] \right|
\geq t C(\ths)/4                    
\]
with probability at least  $1- 2 \beta_1 \exp \left(  -  \frac{ t^2 m \beta_2  }{64 M^2} \right)$.
Using union bounds we estimate the numerator of the second expression in \eqref{der1}.
This fact and \eqref{der2} imply that for every positive $t$ and natural $n,m$ with probability at least
$1- \beta_1 \exp \left(  -  \frac{m \beta_2  }{4 M^2}                     \right)  -2 \bbd \beta_1 \exp \left(  -  \frac{ t^2 m \beta_2  }{64 M^2} \right)
$ 
we have 
\begin{equation}
\label{deriv_form2}
\left|\frac{\frac{1}{m} \sum\limits_{k=1}^{m} w_k(\ths) \left[ J(Y^k) - \frac{\grad C(\ths)}{C(\ths)} \right]}
                 {\frac{1}{m} \sum\limits_{k=1}^{m} w_k(\ths) } \right| _\infty \leq t/2\,.
\end{equation}
Taking \eqref{deriv_form} and \eqref{deriv_form2} together we finish the proof.

\end{proof}

\begin{corollary}
\label{cor_der}
Let $\varepsilon >0, \xi >1$ and
\[\lambda = \frac{\xi + 1}{\xi -1}  \max\left(2 \sqrt{\frac{2 \log(2 \bbd/\varepsilon)}{n}}, 
8 M \sqrt{ \frac{\log\left[(2 \bbd+1)\beta_1 /\varepsilon\right]}{m\beta_2}}
\right)\,.
\]
Conditions \eqref{ncond} and \eqref{mcond} imply 
\[P(\Omega_1 ) \geq 1-2 \varepsilon\,.\]
\end{corollary}

\begin{proof}
We take $t=\frac{\xi -1}{\xi +1} \lambda$ in Lemma \ref{derivative}.
\end{proof}

\begin{lemma}
\label{matrix_diff}
For every $n,m$ and positive $t$
\begin{multline}
\label{mat_diff1}
P \left( \left| \hess \llnm (\ths) - \hess \log C(\ths)\right| _ \infty \leq t \right) \geq 1-2 \beta_1 \exp \left(  -  \frac{m \beta_2 }{4 M^2}                     \right)  
\\ - 2 \bbd (\bbd+1)  \beta_1 \exp \left(  -  \frac{ t^2 m \beta_2  }{256 M^2} \right)\,.
\end{multline}
\end{lemma}
\begin{proof}
To estimate the  $l_\infty$-norm of the matrix difference  in \eqref{mat_diff1} we bound $l_\infty$-norms of two matrices:
\begin{equation}
\label{second}
\frac{\sum\limits_{k=1}^{m} w_k(\ths) J(Y^k)^{\otimes 2}}
                 {\sum\limits_{k=1}^{m} w_k(\ths)} - \frac{\hess C(\ths)}{C(\ths)}
= \frac{\frac{1}{m} \sum\limits_{k=1}^{m} w_k(\ths) \left[J(Y^k)^{\otimes 2}
 - \frac{\hess C(\ths)}{C(\ths)} \right]}
                 {\frac{1}{m} \sum\limits_{k=1}^{m} w_k(\ths)}
\end{equation}
and 
\begin{equation}
\label{first}
\left[ \frac{\frac{1}{m} \sum\limits_{k=1}^{m} w_k(\th) J(Y^k)}
                 {\frac{1}{m} \sum\limits_{k=1}^{m} w_k(\th) } \right]^{\otimes 2}
- \left[ \frac{\grad C(\ths)}
                 { C(\ths)} \right]^{\otimes 2}\, .
\end{equation}

The denominator of the right-hand side of \eqref{second} has been estimated in the proof of 
Lemma \ref{derivative}, so we bound the numerator. We can calculate that 
\[
\Ex _{Y \sim h} \left[\frac{\exp \left[(\ths)'J(Y)\right]}{h(Y)} \left( J(Y) ^{\otimes 2} - \frac{\hess C(\ths)}{C(\ths)} \right) \right] =0
\]
and for every $y$ and two pairs of indices $r<s, r'<s'$ we have
\[
\frac{\exp((\ths)'J(y)}{h(y)} \left| J \rs (y) J_{r's'} (y) - \frac{\hess _{rs,r's'} C(\ths)}{C(\ths)}
\right| \leq 2 M C(\ths)\,.
\]
Using the union bound and \eqref{der2} we upper-bound the $l_\infty$-norm of the right-hand side of \eqref{second} by $t/2$ with probability at least 
\[1-\beta_1 \exp \left(  -  \frac{m \beta_2  }{4 M^2}                     \right)  -2 \bbd  ^2 \beta_1 \exp \left(  -  \frac{ t^2 m \beta_2  }{64 M^2} \right)\,.\]

The last step of the proof is handling with the $l_\infty$-norm of \eqref{first}. This expression can be upper-bounded by 
\begin{equation}
\label{matrix1}
\left|\frac{\sum\limits_{k=1}^{m} w_k(\th) J(Y^k)}
                 {\sum\limits_{k=1}^{m} w_k(\th) } 
- \frac{\grad C(\ths)}
                 { C(\ths)}\right| _\infty 
\left(\left|\frac{\sum\limits_{k=1}^{m} w_k(\th) J(Y^k)}
                 {\sum\limits_{k=1}^{m} w_k(\th) } \right| _\infty
+ \left|\frac{\grad C(\ths)}
                 { C(\ths)}\right| _\infty 
\right) _\infty\, .
\end{equation}
The first term in \eqref{matrix1} has been bounded with high probability in the proof of 
Lemma \ref{derivative}. The remaining two can be easily estimated by one. Therefore, for every positive $t$  the $l_\infty$-norm of \eqref{first} is not greater than $t/2$ with probability at least 
\[
1-\beta_1 \exp \left(  -  \frac{m \beta_2  }{4 M^2}                     \right)  -2 \bbd  \beta_1 \exp \left(  -  \frac{ t^2 m \beta_2  }{256 M^2} \right)\, . 
\]
 Putting together this fact and the bound of \eqref{second} we finish the proof.
\end{proof}

\begin{corollary}
\label{Fdiff}
If \eqref{mcond}, then for every $\varepsilon>0$  the following inequality 
\[
\bef \geq \ef 
-  16 \bd0 (1+\xi)^2 M  \sqrt{  \frac{\log\left[2 \bbd (\bbd+1)\beta_1 /\varepsilon\right]}{m\beta_2}
}
\]
has probability at least $1-2 \varepsilon.$
\end{corollary}

\begin{proof} We take 
\[t= 16M  \sqrt{  \frac{\log\left[2 \bbd (\bbd+1)\beta_1 /\varepsilon\right]}{m\beta_2}
}
\]
in Lemma \ref{matrix_diff} and use \citet[Lemma 4.1 (ii)]{Cox13}.  
\end{proof}

\section{Proofs of main results}

\label{proofsmain}

\begin{proof}[Proof of Theorem \ref{main}]
Fix $\varepsilon >0, \xi >1$ and denote $\gamma=\gamma(\xi) = \frac{\alpha(\xi)-2}{\alpha(\xi)} \in (0,1).$ 
First, from Corollary \ref{cor_der} we know that $P (\Omega_1 ) \geq 1-2\varepsilon.$
Using the condition \eqref{mcond} we obtain that 
\[
\ef 
-  16 \bd0 (1+\xi)^2 M  \sqrt{  \frac{\log\left[2 \bbd (\bbd+1)\beta_1 /\varepsilon\right]}{m\beta_2}} \geq \gamma \ef\,.
\]
Therefore, from Corollary \ref{Fdiff} we have that $P (\bef \geq \gamma \ef) \geq 1-2 \varepsilon.$ It is not difficult to calculate that 
\[
\frac{(1+\xi) \bd0 \lambda}{ \gamma \ef} \leq e^{-1}\,,
\]
so we have also $P (\Omega_2 ) \geq 1-2\varepsilon.$
To finish the proof we use Lemma \ref{estim} (with $\eta=1$ for simplicity) and again bound $\bef$ from above by $\gamma \ef $ in the event $A$ defined in \eqref{estim1}.

\end{proof}

\begin{proof}[Proof of Corollary \ref{maincor}]
The proof is a simple consequence of the uniform bound \eqref{main_in} obtained in Theorem \ref{main}. Indeed, for an  arbitrary pair of indices $(r,s) \notin T$ we obtain 
\[
|\hth \rs | =|\hth \rs  - \ths \rs| \leq R_n^m\,,
\] 
so $(r,s) \notin \hat{T}.$
Analogously, if we take a pair  $(r,s) \in T,$ then 
\[
|\hth \rs | \geq  |\ths \rs| -|\hth \rs  - \ths \rs| > 2 \delta  - R_n^m \geq \delta\,. 
\]
\end{proof}

\bibliography{refs}

\end{document}